\documentclass[letterpaper, conference, 10pt]{ieeeconf}
\IEEEoverridecommandlockouts
\usepackage[hypertexnames=false,		% for autonum package
			colorlinks=true,
			linkcolor=black,
			citecolor=black,
			pdfencoding=auto,
			bookmarks=false,
			pdfstartview=FitH]{hyperref}
\usepackage{amsmath, mathrsfs, amssymb, mathtools, autonum}
	\allowdisplaybreaks
	\DeclarePairedDelimiter\norm{\|}{\|}

	\providecommand\given{}
	\newcommand\SetSymbol[1][]{\nonscript\;#1\vert\nonscript\;
	\mathopen{}\allowbreak}
	\DeclarePairedDelimiterX\set[1]\{\}{%
	\renewcommand\given{\SetSymbol[\delimsize]}
	#1
	}

\usepackage{txfonts}
	\DeclareMathAlphabet{\mathcal}{OMS}{cmsy}{m}{n}
\usepackage{graphicx}
	\graphicspath{{fig/}}
\usepackage[percent]{overpic}
\usepackage{color}
	\definecolor{mblue}{RGB}{0,0,205}
\usepackage{subfigure}
\usepackage{tikz}
	\usetikzlibrary{shapes,arrows}
	\tikzstyle{frame} = [draw, -latex]
	\tikzstyle{line} = [draw, line width=0.5mm]
	\tikzstyle{line2} = [draw, dashdotted]
	\tikzstyle{line3} = [draw, dashed, line width=0.5mm]
	\tikzstyle{place} = [circle, draw=black, fill=black, thick, inner sep=2pt, minimum size=1mm]
	\tikzstyle{placeRed} = [circle, draw=red, fill=red, thick, inner sep=2pt, minimum size=1mm]
	\tikzstyle{vertex} = [circle, draw=black, fill=black, thick, inner sep=2pt, minimum size=1mm]
\usepackage[noadjust]{cite}
\usepackage{enumerate}
\usepackage{url}
\usepackage[linesnumbered,ruled]{algorithm2e}
\usepackage{acronym}
	\acrodef{gir}[GIR]{globally and infinitesimally rigid}
	\acrodef{mgir}[MGIR]{minimally globally and infinitesimally rigid}
	\acrodef{mir}[MIR]{minimally and infinitesimally rigid}
	\acrodef{ir}[IR]{infinitesimally rigid}

\newtheorem{lemma}{\bfseries Lemma}

\newtheorem{remark}{\bfseries Remark}
\newtheorem{theorem}{\bfseries Theorem}

\newcommand{\myemph}{\emph}
\newcommand{\myfig}{Fig.~}

\title{\LARGE \bf
Distance-based Control of K4 Formation with Almost Global Convergence
}

\author{
Myoung-Chul Park\authorrefmark{2},
Zhiyong Sun\authorrefmark{3},
Minh Hoang Trinh\authorrefmark{2},
Brian D. O. Anderson\authorrefmark{3},
and Hyo-Sung Ahn\authorrefmark{2}
\thanks{${}^{\dag}$School of Mechatronics, Gwangju Institute of Science and Technology (GIST), 123 Cheomdan-gwagiro, Buk-gu, Gwangju, 61005 Republic of Korea. E-mail: {\tt\small\{mcpark,trinhhoangminh,hyosung\}@gist.ac.kr}}
\thanks{${}^{\ddag}$Zhiyong Sun is with Shandong Computer Science Center (SCSC), Jinan, China; Brian D. O. Anderson was a visiting expert with SCSC. Zhiyong Sun and Brian D. O. Anderson are with National ICT Australia and Research School of Engineering, The Australian National University, Canberra ACT 0200, Australia. {\tt\small\{zhiyong.sun,brian.anderson\}@anu.edu.au}}
}

\begin{document}
\maketitle
\thispagestyle{empty}
\pagestyle{empty}

%%%%%%%%%%%%%%%%%%%%%%%%%%%%%%%%%%%%%%%%%%%%%%%%%%%%%%%%%%%%%%%%%%%%%%%%

\begin{abstract}
In this paper, we propose a distance-based formation control strategy that can enable four mobile agents, which are modelled by a group of single-integrators, to achieve the desired formation shape specified by using six consistent inter-agent distances in a 2-dimensional space. The control law is closely related to a gradient-based control law formed from a potential function reflecting the error between the actual inter-agent distances and the desired inter-agent distances. There are already control strategies achieving the same objective in a distance-based control manner in the literature, but the results do not yet include a global as opposed to local stability analysis. We propose a control strategy modified from the existing gradient-based control law so that we can achieve almost global convergence to the desired formation shape, and the control law uses known properties for an associated formation shape control problem involving a four-agent tetrahedron formation in 3-dimensional space. Simulation results verifying our analysis are also presented.
\end{abstract}

%%%%%%%%%%%%%%%%%%%%%%%%%%%%%%%%%%%%%%%%%%%%%%%%%%%%%%%%%%%%%%%%%%%%%%
\section{Introduction}
Formation control of mobile agents has attracted a lot of attention in the field of multi-agent systems recently, and many approaches to handle the control problem of multi-agent formation have been proposed.
For example, a formation control strategy using relative position constraints for single-integrator modeled agents can formulate the formation control problem as a first-order consensus problem; then, under certain graph structures\footnote{For instance, having a rooted directed spanning tree for time-invariant graph cases and being uniformly connected for time-varying graph cases.}, global asymptotic stability of desired formation shapes can be guaranteed by virtue of the convergence results of the consensus algorithms \cite{C:Ren:CSM2007, C:OlfatiSaber:IEEE2007}.
Meanwhile, some other control strategies using inter-agent distance constraints are studied in the literature \cite{C:OlfatiSaber:IFAC2002, C:Krick:IJC2009, C:Yu:SIAMJCO2009, C:Oh:Automatica2011, C:Summers:TAC2011}.
Since desired formations are characterized by relative position vectors of the agents in displacement-based control, the agents need to be equipped with aligned local reference frames, but in distance-based control, the agents can use fully independent local reference frames.
Unlike displacement-based control, many publications on distance-based control of multi-agent formations deal with local stability analysis only with a few exceptions (e.g., \cite{C:Anderson:IFAC2010, C:Dasgupta:AUCC2011, C:Park:CDC2014, C:Sun:CDC2015}) handling global stability issues.

One of the distance-based formation control algorithms is a gradient-based control law formed from a potential function reflecting the error between the actual inter-agent distances and the desired inter-agent distances, where the local asymptotic stability under the control law is explored in \cite{C:Krick:IJC2009} and \cite{C:Oh:IJRNC2014}.
With the same control law used in \cite{C:Krick:IJC2009}, the results in \cite{C:Anderson:IFAC2010} and \cite{C:Dasgupta:AUCC2011} show that, for any four-agent formation that is represented by the four-vertex complete graph $\mathcal{K}_{4}$ in the two-dimensional space and thus termed a $\mathcal{K}_{4}$ formation, any \myemph{rectangular} incorrect equilibrium formation\footnote{An incorrect equilibrium formation is one where the potential function has a critical point not corresponding to the correct formation shape.} is unstable.
In the case of a general $\mathcal{K}_{4}$ formation in the three-dimensional space, i.e., the tetrahedral formation, it is shown that any incorrect equilibrium formation is unstable \cite{C:Park:CDC2014}.
Furthermore, \cite{C:Sun:CDC2015} shows more generalized results such that, for the case of formations in an arbitrarily finite dimensional space with a \myemph{realizable} formation shape, any \myemph{degenerate}\footnote{By degenerate formation is meant one for which the space spanned by the vertex positions is of lesser dimension than that of the space in which the formation is supposed to exist.} incorrect equilibrium formation is unstable.

Nevertheless, it is still not known whether the gradient control law may have a locally attractive incorrect equilibrium formation which could be nondegenerate for general cases, e.g., for a general rather than rectangular $\mathcal{K}_{4}$ formation in the two-dimensional space.
However, numerical simulations show that we can observe the existence of an attractive incorrect equilibrium formation for general cases under the gradient control law.
For instance, \myfig\ref{Fig:5AMIRCorrectTrj} shows convergence to a desired formation shape in the two-dimensional space with 5-agent system under the gradient control law proposed in \cite{C:Krick:IJC2009}, but \myfig\ref{Fig:5AMIRIncorrectTrj} shows convergence to an incorrect nondegenerate formation shape from another initial condition under the same control law.
\begin{figure}[!t]
\centering
\subfigure[Convergence to a desired formation shape.]{
\includegraphics[width=0.48\textwidth]{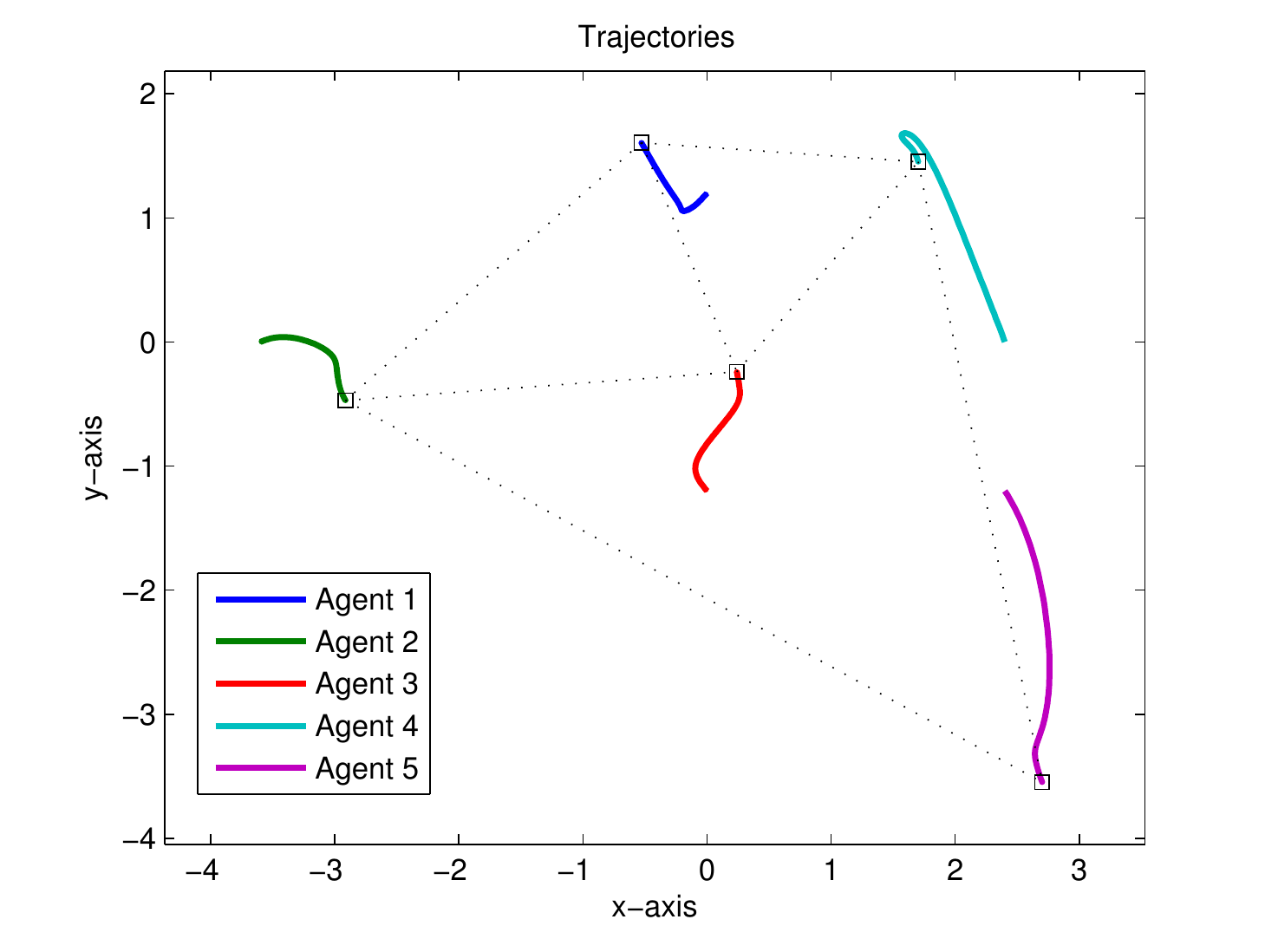}
\label{Fig:5AMIRCorrectTrj}
}
\subfigure[Convergence to an incorrect equilibrium formation.]{
\includegraphics[width=0.48\textwidth]{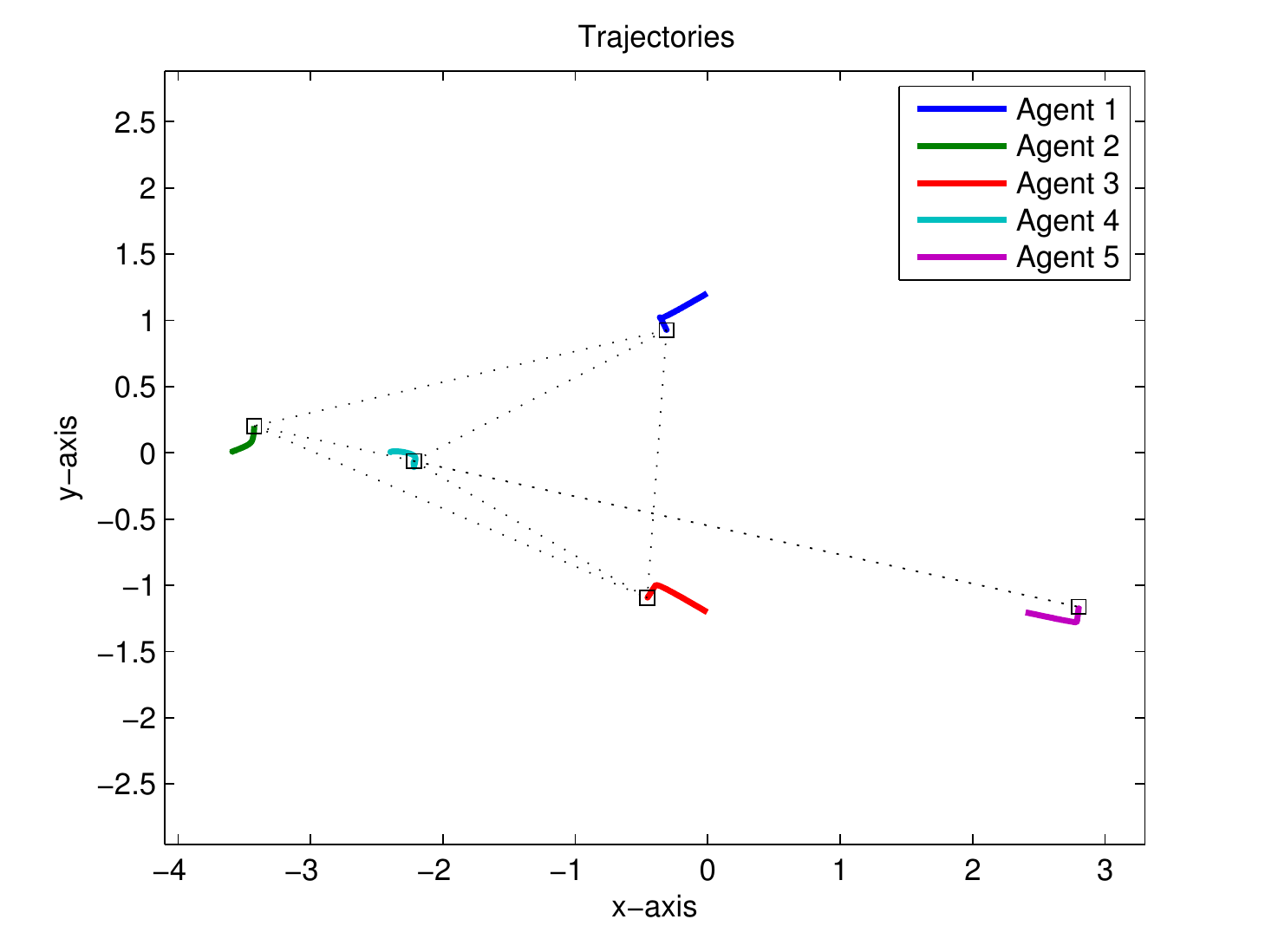}
\label{Fig:5AMIRIncorrectTrj}
}
\caption{The agents are governed by the control law proposed in \cite{C:Krick:IJC2009}. The dotted lines represent the edges, and the square marks denote the final locations of the agents. The squared inter-agent distances for the desired formation shape are 10, 4, 5, 10, 41, 5, and 26 for edges $\set{1,2}$, $\set{1,3}$, $\set{1,4}$, $\set{2,3}$, $\set{2,5}$, $\set{3,4}$, and $\set{4,5}$, respectively.}
\label{Fig:5AMIR}
\end{figure}

The example formation used in \myfig\ref{Fig:5AMIR} is a general formation in two-dimensional space not largely different to a $\mathcal{K}_{4}$ formation with respect to the number of agents and the number of edges, and more agents and edges again may result in a larger number of attractive incorrect equilibrium formations associated with a single correct equilibrium.
On the other hand a $\mathcal{K}_{3}$ formation, i.e., triangular formation, in two-dimensional space has only degenerate incorrect equilibrium formations, which are unstable from the consequence of \cite{C:Sun:CDC2015}.
Consequently a key remaining question in relation is to establish whether or not there is an attractive incorrect equilibrium formation for the $\mathcal{K}_{4}$ formation beyond the results in \cite{C:Anderson:IFAC2010, C:Dasgupta:AUCC2011, C:Summers:MSC2013}.

In this paper, rather than dealing with this question we propose a modified control law to achieve desired formation shapes for $\mathcal{K}_{4}$ formation in the two-dimensional space.
The proposed control law makes the $\mathcal{K}_{4}$ formation in the two-dimensional space mimic the tetrahedral formation evolving in the three-dimensional space, thereby we can use the characteristics of $\mathcal{K}_{4}$ formation in the three-dimensional space observed in \cite{C:Park:CDC2014} and \cite{C:Sun:CDC2015} so that we can guarantee almost global convergence to desired formation shapes.

The rest of the paper is organized as follows.
In Section~\ref{Sec:preliminary}, we provide background knowledge on $\mathcal{K}_{4}$ formations in two- and three-dimensional spaces under the conventional gradient control law.
We propose a modified control law in Section~\ref{SSec:K4FormationInReduced3D} with stability analysis.
Simulation results with the proposed control law are shown in Section~\ref{Sec:simulation}, and we summarize the paper in Section~\ref{Sec:Conclusion}.

%%%%%%%%%%%%%%%%%%%%%%%%%%%%%%%%%%%%%%%%%%%%%%%%%%%%%%%%%%%%%%%%%%%%%%
\section{Background knowledge}
\label{Sec:preliminary}
In the rest of the paper, we are going to use the following
notation:
\begin{itemize}
\item $\mathbb{R}^{n}$: $n$-dimensional Euclidean space
\item $\norm{x}$: the Euclidean norm of a vector $x$
\end{itemize}

Consider mobile agents that can move on a flat surface so that their locations can be represented by 2-vectors.
We assume that any two neighboring agents can measure the relative position between them, but they do not know their absolute positions.
The agents do not need to have the same coordinate frame.
Thus, each agent can measure the relative positions of its neighbors in its own local coordinate frame.
Under such a description, the agents and the pairs of agents which are neighbors can be represented by a graph (agents corresponding to vertices and neighbor pairs corresponding to edges) which is \myemph{realized} in $\mathbb{R}^{2}$ as a formation.
Thus we introduce relevant notation and terminologies in the following subsections.

%%%%%%%%%%%%%%%%%%%%%%%%%%%%%%%%%%%%%%
\subsection{Formation graph}
Let $\mathcal{G}$ denote a graph defined by $\mathcal{G} = (\mathcal{V},\mathcal{E})$, where $\mathcal{V} = \set{1,\ldots,N}$ is the set of all vertices representing the agents, and $\mathcal{E} = \set{\ldots,\set{i,j},\ldots}$ is the set of all edges representing some pairs of the vertices.
Let $p_i \in \mathbb{R}^{n}$ denote the position vector of vertex~$i$.
We call $p = [p_1^{\top}~~\ldots~~p_{N}^{\top}]^{\top} \in \mathbb{R}^{Nn}$ a \myemph{realization} of $\mathcal{G}$ in $\mathbb{R}^{n}$.
A \myemph{framework} is defined by a pair of graph and its realization, which is denoted by $(\mathcal{G},p)$.
Two realizations $p$ and $p'$ are said to be \myemph{congruent} if $\norm{p_i - p_j} = \norm{p_i' - p_j'}$ for all $i,j \in \mathcal{V}$.

%%%%%%%%%%%%%%%%%%%%%%%%%%%%%%%%%%%%%%
\subsection{$\mathcal{K}_{4}$ formation in $\mathbb{R}^{2}$}
\label{SSec:K4FormationIn2D}
Consider four mobile agents moving on a plane, where the corresponding formation graph is given by $\mathcal{G} = (\mathcal{V},\mathcal{E})$ with $\mathcal{V} = \set{1,2,3,4}$ and $\mathcal{E} = \set*{\set{i,j} \given 1 \le i < j \le 4}$.
Let $p_i = [p_{ix}~~p_{iy}]^{\top} \in \mathbb{R}^{2}$ denote the position vector of agent~$i$ for all $i \in \mathcal{V}$.
We call $(\mathcal{G},p)$ a $\mathcal{K}_{4}$ formation in $\mathbb{R}^{2}$ because the underlying formation graph is the complete graph with four vertices, and the graph is supposed to be realized in $\mathbb{R}^{2}$.

We assume that the agents are governed by
\begin{align}
\dot{p}_i = u_i, \ \forall i \in \mathcal{V},
\end{align}
where $u_i = [u_{ix}~~u_{iy}]^{\top} \in \mathbb{R}^{2}$ is the control input for agent~$i$.
Achieving the desired formation shape by distance-based control means achieving the desired inter-agent distances so that
\begin{align}
\norm{p_i - p_j} \to d_{ij}, \ \forall \set{i,j} \in \mathcal{E}.
\end{align}
Also, normally, we need to guarantee that the agents converge to some locations for which the desired inter-agent distances are satisfied.
We assume that the distances $d_{ij} \in \mathbb{R}_{>0}$, $\set{i,j} \in \mathcal{E}$, are realizable in $\mathbb{R}^{2}$, which means that there exists a realization $\bar{p} = [p_1^{\top}~~\ldots~~p_4^{\top}]^{\top} \in \mathbb{R}^{8}$ such that $\norm{\bar{p}_i - \bar{p}_j} = d_{ij}$ for all $\set{i,j} \in \mathcal{E}$.

Note that the realizability of a nondegenerate $\mathcal{K}_{4}$ formation in $\mathbb{R}^{2}$ with given six distances can be checked by some triangle inequalities and Cayley-Menger determinant $\det C$, where
\begin{align}
C =
\begin{bmatrix}
0 &d_{12}^{2} &d_{13}^{2} &d_{14}^{2} &1\\
d_{12}^{2} &0 &d_{23}^{2} &d_{24}^{2} &1\\
d_{13}^{2} &d_{23}^{2} &0 &d_{34}^{2} &1\\
d_{14}^{2} &d_{24}^{2} &d_{34}^{2} &0 &1\\
1 &1 &1 &1 &0
\end{bmatrix}.
\end{align}
Thus, each of the triples $(d_{12},d_{13},d_{23})$, $(d_{12},d_{14},d_{24})$, $(d_{13},d_{14},d_{34})$, and $(d_{23},d_{24},d_{34})$ has to satisfy the triangle inequality, and $\det C$ must be zero for the realizability of the $\mathcal{K}_{4}$ formation in $\mathbb{R}^{2}$ \cite{C:Wirth:EM2009}.
If $\det C > 0$, provided that the triangle inequalities are satisfied, then the $\mathcal{K}_{4}$ formation is realizable in $\mathbb{R}^{3}$.

Let $v \colon \mathbb{R}^{8} \to \mathbb{R}_{\ge 0}$ be a potential function defined by
\begin{align}
v(p) = \frac{1}{4}\sum_{\set{i,j} \in \mathcal{E}}\left(\norm{p_i - p_j}^2 - d_{ij}^2\right)^2,\label{EQ:potentialFnK4inR2}
\end{align}
and consider a control law given by
\begin{align}
\label{EQ:gradControl2D}
\dot{p} = u = -\left[\frac{\partial v}{\partial p}\right]^{\top},
\end{align}
where $u = [u_1^{\top}~~\ldots~~u_4^{\top}]^{\top} \in \mathbb{R}^{8}$.
Then, we can consider two equilibrium sets of which the union is the entire equilibrium set.
One is the desired correct equilibrium set $\mathcal{P}_{\mathrm{d}}$ defined by
\begin{align}
\mathcal{P}_{\mathrm{d}} = \set*{p \in \mathbb{R}^{8} \given v(p) = 0},
\end{align}
and the other is the incorrect equilibrium set $\mathcal{P}_{\mathrm{i}}$ given by
\begin{align}
\mathcal{P}_{\mathrm{i}} = \set*{p \in \mathbb{R}^{8} \given \frac{\partial v}{\partial p} = 0,\ v(p) \ne 0}.
\end{align}
Thus, $\mathcal{P}_{\mathrm{d}}$ and $\mathcal{P}_{\mathrm{i}}$ partition the equilibrium set $\mathcal{P}_{\mathrm{eq}}$ defined by
\begin{align}
\mathcal{P}_{\mathrm{eq}} = \set*{p \in \mathbb{R}^{8} \given \frac{\partial v}{\partial p} = 0}.
\end{align}
The $\mathcal{K}_{4}$ formation in $\mathbb{R}^{2}$ is a particular example of a general \myemph{rigid} formation in $\mathbb{R}^{2}$.
For those general formations, it is shown in \cite{C:Krick:IJC2009} that the negative gradient of a potential function (which can be extended from \eqref{EQ:potentialFnK4inR2} for more agents and edges) guarantees, under a single-integrator model, local asymptotic stability of the desired formation shape if the desired formation shape is infinitesimally rigid, and it is further revealed that such an infinitesimal rigidity condition can be relaxed to rigidity condition \cite{C:Oh:IJRNC2014}.
Thus, we can conclude that $\mathcal{P}_{\mathrm{d}}$ is asymptotically stable under \eqref{EQ:gradControl2D}.

On the other hand, an example of a $\mathcal{K}_{4}$ formation is explored in \cite{C:Krick:IJC2009} where the desired distances are given by
\begin{align}
d_{12} = 1,\ d_{13} = \sqrt{5},\ d_{14} = 2,\ d_{23} = 2,\ d_{24} = \sqrt{5},\ d_{34} = 1,
\end{align}
which define a rectangular formation.
The paper creates the impression that there exist incorrect equilibria which may be attractive under \eqref{EQ:gradControl2D}, but it is shown in \cite{C:Anderson:IFAC2010} and \cite{C:Dasgupta:AUCC2011} that actually this cannot be the case, and the equilibrium in question is a saddle point.
Furthermore, \cite{C:Anderson:IFAC2010} and \cite{C:Dasgupta:AUCC2011} show that, for $\mathcal{K}_{4}$ formations in $\mathbb{R}^{2}$, any incorrect equilibrium formation associated with a rectangular desired formation is unstable (i.e. is a saddle point or is completely unstable) under \eqref{EQ:gradControl2D}.

It can be shown that, under \eqref{EQ:gradControl2D}, $p(t)$ approaches $\mathcal{P}_{\mathrm{eq}}$ as $t \to \infty$ from the following argument.
Consider $v(p)$ defined in \eqref{EQ:potentialFnK4inR2}.
By taking the derivative of $v$, we have
\begin{align}
\dot{v} = -\norm*{\frac{\partial v}{\partial p}}^{2} \le 0.
\end{align}
Since $v$ is lower bounded and $\dot{v}$ is non-positive, $v$ converges to a constant.
Then, from Barbalat's lemma \cite[Lemma~8.2]{C:Khalil:Book2002}, we can conclude that $\dot{v}$ converges to 0 as $t \to \infty$, which means that $p$ approaches $\mathcal{P}_{\mathrm{eq}}$ as $t \to \infty$.

From the preceding argument, we can conclude that if $\mathcal{P}_{\mathrm{i}}$ is unstable, then $\mathcal{P}_{\mathrm{d}}$ is almost globally attractive in the sense that, for almost all initial conditions, $p(t)$ approaches $\mathcal{P}_{\mathrm{d}}$ as $t \to \infty$.
Moreover, since it is known that $p(t)$ converges to a point in $\mathcal{P}_{\mathrm{d}}$ if $p(t)$ approaches $\mathcal{P}_{\mathrm{d}}$ \cite{C:Krick:IJC2009}, we can conclude that $p(t)$ almost globally converges to $\mathcal{P}_{\mathrm{d}}$.
However, it is still an open question whether $\mathcal{P}_{\mathrm{i}}$ is unstable for general $\mathcal{K}_{4}$ formations with desired distances corresponding to non-rectangular formations so the almost global convergence of $p(t)$ to $\mathcal{P}_{\mathrm{d}}$ is still questionable.

Instead of trying to show the instability of incorrect equilibrium formations of \eqref{EQ:gradControl2D}, in this paper, we are going to propose a modified control law to guarantee the almost global attractiveness of the desired formations by taking advantage of the properties of $\mathcal{K}_{4}$ formations in $\mathbb{R}^{3}$.

%%%%%%%%%%%%%%%%%%%%%%%%%%%%%%%%%%%%%%
\subsection{$\mathcal{K}_{4}$ formation in $\mathbb{R}^{3}$}
\label{SSec:K4FormationIn3D}
Instead of $\mathcal{K}_{4}$ formations in $\mathbb{R}^{2}$, let us temporarily consider a $\mathcal{K}_{4}$ formation in $\mathbb{R}^{3}$.
Some properties of $\mathcal{K}_{4}$ formations in $\mathbb{R}^{3}$ are going to be used in Section~\ref{SSec:stabilityOfLockedK4Formation}.

Let $P_i = [P_{ix}~~P_{iy}~~P_{iz}]^{\top} \in \mathbb{R}^{3}$.
$P = [P_1^{\top}~~\ldots~~P_4^{\top}]^{\top} \in \mathbb{R}^{12}$.
Analogously to the 2-D $\mathcal{K}_{4}$ formation in Section~\ref{SSec:K4FormationIn2D}, we assume that the agents are governed by
\begin{align}
\dot{P}_i = U_i,\ \forall i \in \mathcal{V},
\end{align}
where $U_i = [U_{ix}~~U_{iy}~~U_{iz}]^{\top} \in \mathbb{R}^{3}$ is the control input for agent~$i$.
Let $D_{ij} \in \mathbb{R}_{>0}$, $\set{i,j} \in \mathcal{E}$, be the distances induced from the desired 3-D $\mathcal{K}_{4}$ formation.
Let $V \colon \mathbb{R}^{12} \to \mathbb{R}_{\ge 0}$ be a potential function defined by
\begin{align}
V(P) = \frac1{4}\sum_{1 \le i < j \le 4}\left(\norm{P_i - P_j}^2 - D_{ij}^2\right)^2.
\end{align}
Then, consider the following systems:
\begin{align}
\label{EQ:EQ:gradControl3D}
\dot{P} = U = -\left[\frac{\partial V}{\partial P}\right]^{\top},
\end{align}
where $U = [U_1^{\top}~~\ldots~~U_4^{\top}]^{\top} \in \mathbb{R}^{12}$.

\begin{lemma}[\cite{C:Park:CDC2014, C:Sun:CDC2015}]
\label{Lemma:negEigValOfHV}
Let $H_{V}(P)$ denote the Hessian matrix of $V(P)$.
For any incorrect equilibrium point $P^*$ of \eqref{EQ:EQ:gradControl3D}, $H_{V}(P^*)$ has a negative eigenvalue.
\end{lemma}

Note that $H_{V}$ in Lemma~\ref{Lemma:negEigValOfHV} is the Jacobian matrix of the right side of \eqref{EQ:EQ:gradControl3D}.
Hence, the existence of a negative eigenvalue of $H_{V}$ at $P^*$ implies the existence of a positive eigenvalue of the Jacobian matrix at $P^*$, which means that $P^*$ is an unstable equilibrium point of \eqref{EQ:EQ:gradControl3D}.
Also, Lemma~\ref{Lemma:negEigValOfHV} means that $V$ does not have a local minimum at any incorrect equilibrium point of \eqref{EQ:EQ:gradControl3D}.

%%%%%%%%%%%%%%%%%%%%%%%%%%%%%%%%%%%%%%%%%%%%%%%%%%%%%%%%%%%%%%%%%%%%%%
\section{$\mathcal{K}_{4}$ formation locked on $x$-$y$ plane}
%%%%%%%%%%%%%%%%%%%%%%%%%%%%%%%%%%%%%%
\subsection{Equations of motion and modified control law}
\label{SSec:K4FormationInReduced3D}
In this section, we are considering a $\mathcal{K}_{4}$ formation which exists in a 2-dimensional ambient space in reality, but one of the agents is assumed to be in a \emph{virtual} 3-dimensional space containing the 2-dimensional space where all other three agents live.
For example, consider the formation illustrated in \myfig\ref{Fig:K4FrameworkinR2}.
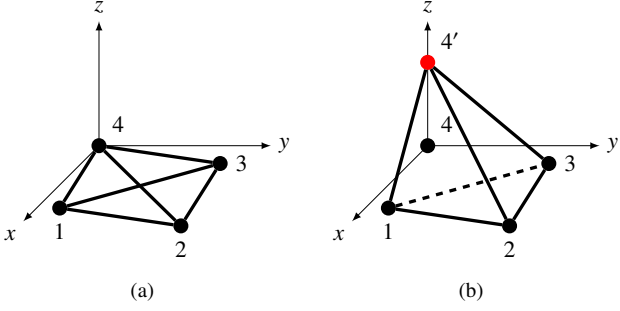
\begin{figure}
\centering
\subfigure[\label{Fig:K4FrameworkinR2}]{
\resizebox{0.225\textwidth}{!}{%
\begin{tikzpicture}[scale=0.7]
\node (py) at (4,0,0) {$y$};
\node (pz) at (0,3,0) {$z$};
\node (px) at (0,0,5) {$x$};

\draw (0,0,0) [frame] -- (px);
\draw (0,0,0) [frame] -- (py);
\draw (0,0,0) [frame] -- (pz);

\node[place] (node1) at (0.5,0,3.5) [label=below:$1$] {};
\node[place] (node2) at (3.5,0,4.5) [label=below:$2$] {};
\node[place] (node3) at (3,0,1) [label=right:$3$] {};
\node[place] (node4) at (0,0,0) [label=above right:$4$] {};

\draw (node1) [line] -- node [left] {} (node2);
\draw (node1) [line] -- node [right] {} (node3);
\draw (node1) [line] -- node [right] {} (node4);
\draw (node2) [line] -- node [below] {} (node3);
\draw (node2) [line] -- node [below] {} (node4);
\draw (node3) [line] -- node [below] {} (node4);
\end{tikzpicture}}
}
\subfigure[\label{Fig:K4FrameworkinR3}]{
\resizebox{0.225\textwidth}{!}{%
\begin{tikzpicture}[scale=0.7]
\node (py) at (4,0,0) {$y$};
\node (pz) at (0,3,0) {$z$};
\node (px) at (0,0,5) {$x$};

\draw (0,0,0) [frame] -- (px);
\draw (0,0,0) [frame] -- (py);
\draw (0,0,0) [frame] -- (pz);

\node[place] (node1) at (0.5,0,3.5) [label=below:$1$] {};
\node[place] (node2) at (3.5,0,4.5) [label=below:$2$] {};
\node[place] (node3) at (3,0,1) [label=right:$3$] {};
\node[place] (node4) at (0,0,0) [label=above right:$4$] {};
\node[placeRed] (node4p) at (0,1.8,0) [label=above right:$4'$] {};

\draw (node1) [line] -- node [left] {} (node2);
\draw (node1) [line3] -- node [right] {} (node3);
\draw (node1) [line] -- node [right] {} (node4p);
\draw (node2) [line] -- node [below] {} (node3);
\draw (node2) [line] -- node [below] {} (node4p);
\draw (node3) [line] -- node [below] {} (node4p);
\end{tikzpicture}}
}
\caption{Typical $\mathcal{K}_{4}$ formations: (a) represents the actual 2-D formation we are handling, and (b) represents the virtual 3-D formation such that the projection onto the $x$-$y$ plane is equal to the formation in (a).}
\label{Fig:K4Framework}
\end{figure}
Four agents move in the $x$-$y$ plane, and the motion of the agents are basically governed by the equations in Section~\ref{Sec:preliminary}.
However, we introduce a virtual variable enabling an agent to be in the virtual 3-dimensional space by treating the virtual variable as if it is the $z$ coordinate of the position vector of the agent.
Hence, in \myfig\ref{Fig:K4FrameworkinR3}, the actual location of agent~4 is the location of the vertex denoted by 4, but we think as if agent~4 is at the location of the vertex denoted by 4' so the formation can be treated as a virtual tetrahedron formation with its base (defined by the plane containing agents 1,2 and 3) locked on the $x$-$y$ plane.
We assume that the virtual variable (the $z$-coordinate of agent 4' and denoted by $p_{4z}$) can be transmitted to the other agents by communication.

We define a new state vector $q \in \mathbb{R}^{9}$ by
\begin{align}
q
&= \begin{bmatrix}p_1^{\top} &p_2^{\top} &p_3^{\top} &p_4^{\top} &p_{4z}\end{bmatrix}^{\top}\\
&= \begin{bmatrix}p_{1x}&p_{1y}&p_{2x}&p_{2y}&p_{3x}&p_{3y}&p_{4x}&p_{4y}&p_{4z}\end{bmatrix}^{\top},
\end{align}
to represent the overall system of the locked tetrahedron formation. 
Although we are considering the virtual tetrahedron formation still the ultimate control objective is to achieve that
\begin{align}
\label{EQ:achieving2Dformation}
\sqrt{(p_{ix} - p_{jx})^2 + (p_{iy} - p_{jy})^2} \to d_{ij},\ \forall \set{i,j} \in \mathcal{E}.
\end{align}
Equivalently, we want to achieve a desired tetrahedron formation of which the projection onto the $x$-$y$ plane is the same as the desired $\mathcal{K}_{4}$ formation mentioned in Section~\ref{SSec:K4FormationIn2D}.

This can be done through setting another desired distance set for the tetrahedron, as follows.
Let $\alpha > 0$ be arbitrary, and define a tetrahedron formation in $\mathbb{R}^{3}$ with distances $D_{ij}$ given by
\begin{align}
D_{ij} &= d_{ij},\ \forall 1\le i<j \le 3,\\
D_{i4} &= \sqrt{d_{i4}^2 + \alpha^2},\ \forall i \in \set{1,2,3}.
\end{align}
Now, consider the following potential function:
\begin{align}
\bar{V}(q)
&= \frac1{4}\sum_{1\le i< j \le 3}\left[(p_{ix}-p_{jx})^{2} + (p_{iy}-p_{jy})^{2} - D_{ij}^{2}\right]^{2}\\
&\phantom{=\,}+ \frac1{4}\sum_{1 \le i \le 3} \left[(p_{ix}-p_{4x})^{2} + (p_{iy}-p_{4y})^{2} + p_{4z}^{2} - D_{i4}^{2}\right]^2.
\end{align}
We can observe that if $\bar{V}(q) \to 0$, then we achieve \eqref{EQ:achieving2Dformation}, i.e., we can achieve the 2-dimensional desired formation shape.

Suppose that $p_{4z}$ is governed by
\begin{align}
\label{EQ:dynamicsOfp4z}
\dot{p}_{4z} = u_{4z},\ p_{4z}(0) \ne 0,
\end{align}
and \eqref{EQ:dynamicsOfp4z} can be updated numerically by agent~4.
The initial condition $p_{4z}(0)$ can be set as an arbitrary nonzero number (preferably as $\alpha$).
Then, we propose the following control law:
\begin{align}
\label{EQ:reducedDynamics}
\dot{q} = \begin{bmatrix}u^{\top} &u_{4z}\end{bmatrix}^{\top} = -\left[\dfrac{\partial \bar{V}}{\partial q}\right]^{\top}.
\end{align}

Let $e_{ij} = (p_{ix}-p_{jx})^{2} + (p_{iy}-p_{jy})^{2} - D_{ij}^{2}$ for all $1 \le i<j \le 3$ and $e_{i4} = (p_{ix}-p_{4x})^{2} + (p_{iy}-p_{4y})^{2} + p_{4z}^{2} - D_{i4}^{2}$ for all $i \in \set{1,2,3}$.
Then, \eqref{EQ:reducedDynamics} can be written in detail as
\begin{align}
\dot{p}_{1} &= (p_{2}-p_{1})e_{12} + (p_{3}-p_{1})e_{13} + (p_{4}-p_{1})e_{14},\\
\dot{p}_{2} &= (p_{1}-p_{2})e_{12} + (p_{3}-p_{2})e_{23} + (p_{4}-p_{2})e_{24},\\
\dot{p}_{3} &= (p_{1}-p_{3})e_{13} + (p_{2}-p_{3})e_{23} + (p_{4}-p_{3})e_{34},\\
\dot{p}_{4} &= (p_{1}-p_{4})e_{14} + (p_{2}-p_{4})e_{24} + (p_{3}-p_{4})e_{34},\\
\dot{p}_{4z} &= (0-p_{4z})e_{14} + (0-p_{4z})e_{24} + (0-p_{4z})e_{34}.
\end{align}
Note that, under \eqref{EQ:reducedDynamics}, the actual locations of the agents on the $x$-$y$ plane are determined by the equations for $\dot{p}_{1},\ldots,\dot{p}_{4}$.
Therefore, the overall virtual formation can be considered as a tetrahedron formation evolving in $\mathbb{R}^{3}$ with three vertices locked on the $x$-$y$ plane.

\begin{remark}
In actual implementation of the control law in \eqref{EQ:reducedDynamics}, observe that each agent can calculate its control input based on the local measurements.
Specifically, $e_{ij}$ is given by
\begin{align}
e_{ij} =
\begin{cases}
\norm{p_i - p_j}^2 - D_{ij}^2, &(1 \le i < j \le 3),\\
\norm{p_i - p_4}^2 + p_{4z}^2 - D_{i4}^2, & (j = 4),
\end{cases}
\end{align}
where the relative position vectors $p_i - p_j$, $1 \le i<j \le 4$, are assumed to be measured by sensors, and $p_{4z}$ is assumed to be transmitted by communication.
Furthermore, for the relative-position measurements, it is \myemph{not} necessary that each agent be equipped with an aligned local reference frame with respect to the $x$-$y$ plane.
In other words, the agents do not have to share a common sense of the north.
\end{remark}

%%%%%%%%%%%%%%%%%%%%%%%%%%%%%%%%%%%%%%
\subsection{Stability analysis}
\label{SSec:stabilityOfLockedK4Formation}
Now we are going to analyze the stability/instability characteristics of different equilibrium formations of \eqref{EQ:reducedDynamics}.
It is not difficult to show that $q(t)$ finally approaches an equilibrium point of \eqref{EQ:reducedDynamics}.
Thus, whether $q(t)$ approaches an incorrect equilibrium point depends on the characteristics of the incorrect equilibria.

Before we go further, we clarify the meaning of congruence.
In a slight abuse of terminology, we say $P$ and $q$ defined in Sections \ref{SSec:K4FormationIn3D} and \ref{SSec:K4FormationInReduced3D}, respectively, are congruent if we have that $\norm{P_i - P_j} = \norm{p_i - p_j}$ for all $1 \le i<j \le 3$ and that $\norm{P_i - P_4} = \sqrt{\norm{p_i - p_4}^2 + p_{4z}^2}$ for all $i \in \set{1,2,3}$.
Since $q$ defines a tetrahedron of which the face defined by vertices 1, 2, and 3 is attached to the $x$-$y$ plane, such an extension is reasonable.

% $ $\newline\newline\newline
% - One to one and onto
% Consider $\mathbb{R}^{12}/\mathcal{C}$ and $\mathbb{R}^{9}/\mathcal{L}$.
% There exists a one-to-one correspondence $\mathbb{R}^{12}/\mathcal{C} \to \mathbb{R}^{9}/\mathcal{L}$, $a \mapsto b$ such that the representative of $a$ and the representative of $b$ are congruent.  
% \begin{lemma}
% Consider $P \in \mathbb{R}^{12}$ and $q \in \mathbb{R}^9$.
% Then, $P$ is a critical point of $V$ if and only if $q$ is a critical point of $\bar{V}$.
% \end{lemma}

\begin{lemma}
\label{Lemma:notLocalMinimizerIffCondition}
Consider two realizations $P \in \mathbb{R}^{12}$ and $q \in \mathbb{R}^{9}$ which are congruent.
Assume that $P$ and $q$ are critical points of $V$ and $\bar{V}$, respectively.
Then, $V$ does not have a local minimum at $P$ if and only if $\bar{V}$ does not have a local minimum at $q$.
\end{lemma}
\begin{proof}
Suppose that $V$ does not have a local minimum at $P$.
Then, it must be true that, for any $\delta > 0$, there exists $P' \in \set{X \in \mathbb{R}^{12} \given \norm{P - X} < \delta}$ such that $V(P') < V(P)$.
Consider arbitrarily small $\bar{\delta} > 0$.
Then there always exists $q' \in \set{x \in \mathbb{R}^{9} \given \norm{q - x} < \bar{\delta}}$ such that $\bar{V}(q') < \bar{V}(q)$ because we can choose $q'$ so that $q'$ and $P'$ are congruent and that $V(P') < V(P)$ with arbitrarily small $\delta > 0$.
Consequently, $\bar{V}$ does not have a local minimum at $q$ if $V$ does not have a local minimum at $P$.

Conversely, suppose that $\bar{V}$ does not have a local minimum at $q$.
Then, for any $\bar{\delta} > 0$, there exists $q' \in \set{x \in \mathbb{R}^{9} \given \norm{q - x} < \bar{\delta}}$ such that $\bar{V}(q') < \bar{V}(q)$.
Now, for arbitrarily small $\delta > 0$, we can always find $P' \in \set{X \in \mathbb{R}^{12} \given \norm{P - X} < \delta}$ such that $V(P') < V(P)$ from the fact that we can take $P'$ so that $P'$ and $q'$ are congruent and that $\bar{V}(q') < \bar{V}(q)$.
Thus, we can conclude that $V$ does not have a local minimum at $P$ if $\bar{V}$ does not have a local minimum at $q$.

Those two arguments above complete the proof.
\end{proof}

In addition to Lemma~\ref{Lemma:notLocalMinimizerIffCondition}, we can more precisely say that $P$ is a maximizer (minimizer, saddle point, respectively) of $V$ if and only if $q$ is a maximizer (minimizer, saddle point, respectively) of $\bar{V}$ provided that $P$ and $q$ are congruent.
Also such a characteristic implies that $P$ is a critical point of $V$ if and only if $q$ is a critical point of $\bar{V}$.

Recall that, in Lemma~\ref{Lemma:notLocalMinimizerIffCondition}, $q$ can be viewed as a tetrahedral formation obtained by locking $P$ on the $x$-$y$ plane.
Thus, the relationship between $V$ and $\bar{V}$ in this paper is similar to the relationship between $V_{N}$ and $V$ in \cite{C:Anderson:SIAMJCO2014}.
In the reference paper, $N$-agent system evolving in one-dimensional space is considered, and $V$ in \cite{C:Anderson:SIAMJCO2014} is the reduced potential function obtained from the nominal potential function $V_{N}$ by assuming that an agent is fixed at the origin, where $V_{N}$ in \cite{C:Anderson:SIAMJCO2014} corresponds to the one-dimensional version of $V$ in this paper.

\begin{theorem}
\label{Thm:instabilityOfRedDynamics}
Any incorrect equilibrium point of \eqref{EQ:reducedDynamics} is unstable.
\end{theorem}
\begin{proof}
To show the instability of an arbitrary incorrect equilibrium point of \eqref{EQ:reducedDynamics}, we can take an approach showing that linearization of \eqref{EQ:reducedDynamics} at the incorrect equilibrium point has a positive eigenvalue.
Meanwhile, the linearization is equal to the negative Hessian matrix of $\bar{V}$.
Thus, showing that the Hessian matrix has a negative eigenvalue at the incorrect equilibrium point is equal to showing that $\bar{V}$ does not have a local minimum at the same point.
We know that, for an incorrect equilibrium point $q^*$ of \eqref{EQ:reducedDynamics}, we can find an incorrect equilibrium point $P^*$ of \eqref{EQ:EQ:gradControl3D} such that $q^*$ and $P^*$ are congruent.
Furthermore, we know that $\bar{V}$ does not have a local minimum at $q^*$ if and only if $V$ does not have a local minimum at $P^*$ from Lemma~\ref{Lemma:notLocalMinimizerIffCondition}.
Note that $H_V(P^*)$ has a negative eigenvalue from Lemma~\ref{Lemma:negEigValOfHV}, which means that $V$ does not have a local minimum at $P^*$.
Therefore, we can conclude that $q^*$ is unstable.
\end{proof}

On the basis of Theorem~\ref{Thm:instabilityOfRedDynamics}, we can addresses the following proposition.

\begin{theorem}
Under the proposed control law in \eqref{EQ:reducedDynamics}, the agents achieve the desired 2-D $\mathcal{K}_{4}$ formation on the $x$-$y$ plane for almost all initial conditions.
\end{theorem}

\section{Simulation}
\label{Sec:simulation}
Consider a set of desired inter-agent distances defining a $\mathcal{K}_{4}$ formation in $\mathbb{R}^{2}$ as follows.
\begin{align}
d_{12}^{2} = 16,\ d_{13}^{2} = 25,\ d_{14}^{2} = 10,\ d_{23}^{2} = 17,\ d_{24}^{2} = 18,\ d_{34}^{2} = 5.\label{EQ:ExDistanceConstraints2D}
\end{align}
A representative realization satisfying the distance constraints in \eqref{EQ:ExDistanceConstraints2D} is given by $\bar{p} = [0~~0~~4~~0~~3~~4~~1~~3]^{\top}$.
Let $\alpha = 1$.
Then, the corresponding distance constraints for a virtual tetrahedron are given by
\begin{align}
D_{12}^{2} &= 16,\ D_{13}^{2} = 25,\ D_{14}^{2} = 11,\\
D_{23}^{2} &= 17,\ D_{24}^{2} = 19,\ D_{34}^{2} = 6.\label{EQ:ExDistanceConstraints3D}
\end{align}

Then, a simulation for \eqref{EQ:reducedDynamics} results in \myfig\ref{Fig:Formation3D}.
\begin{figure}[!t]
\centering
\subfigure[]{
\includegraphics[width=0.48\textwidth]{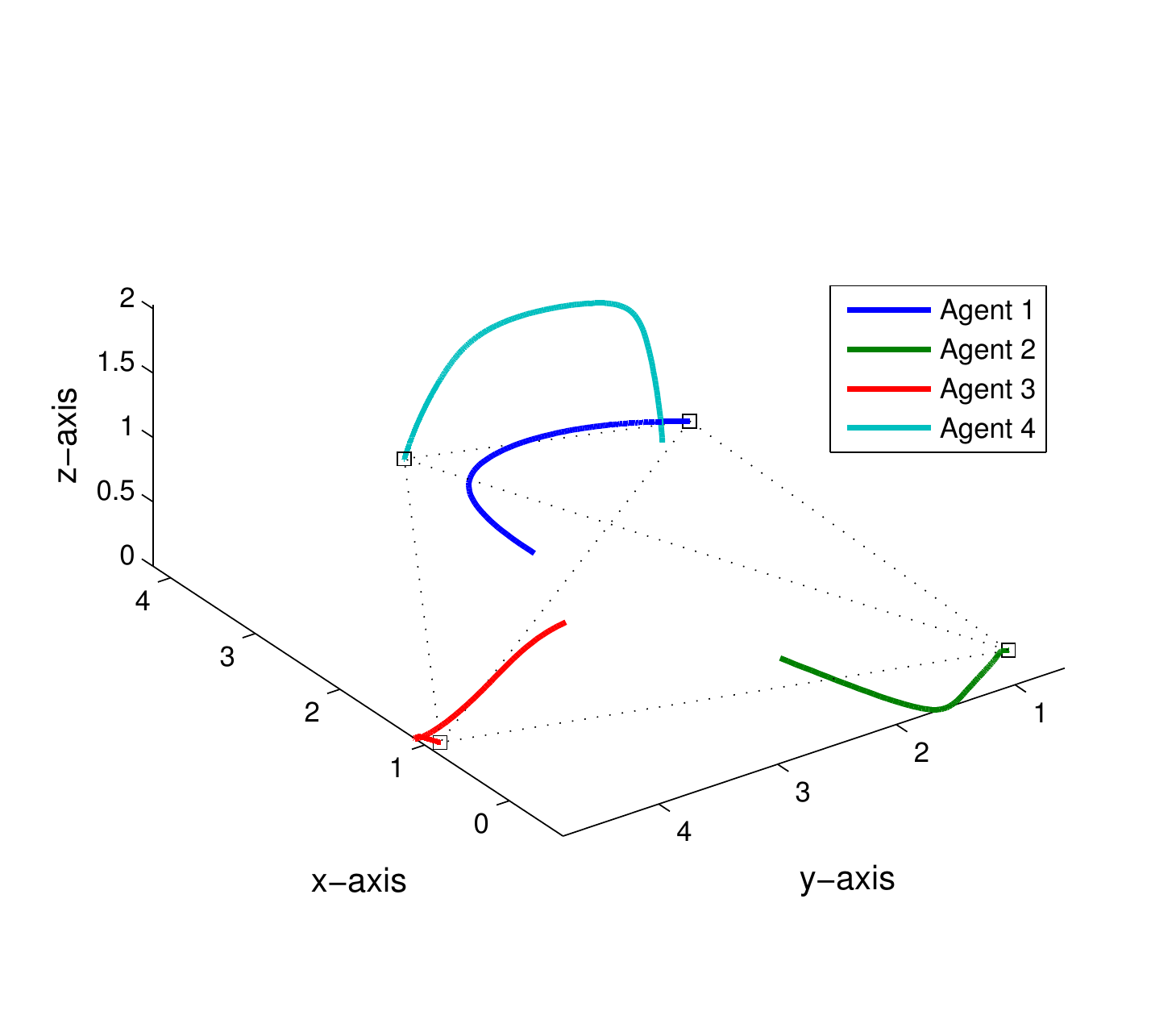}
}
\subfigure[]{
\includegraphics[width=0.48\textwidth]{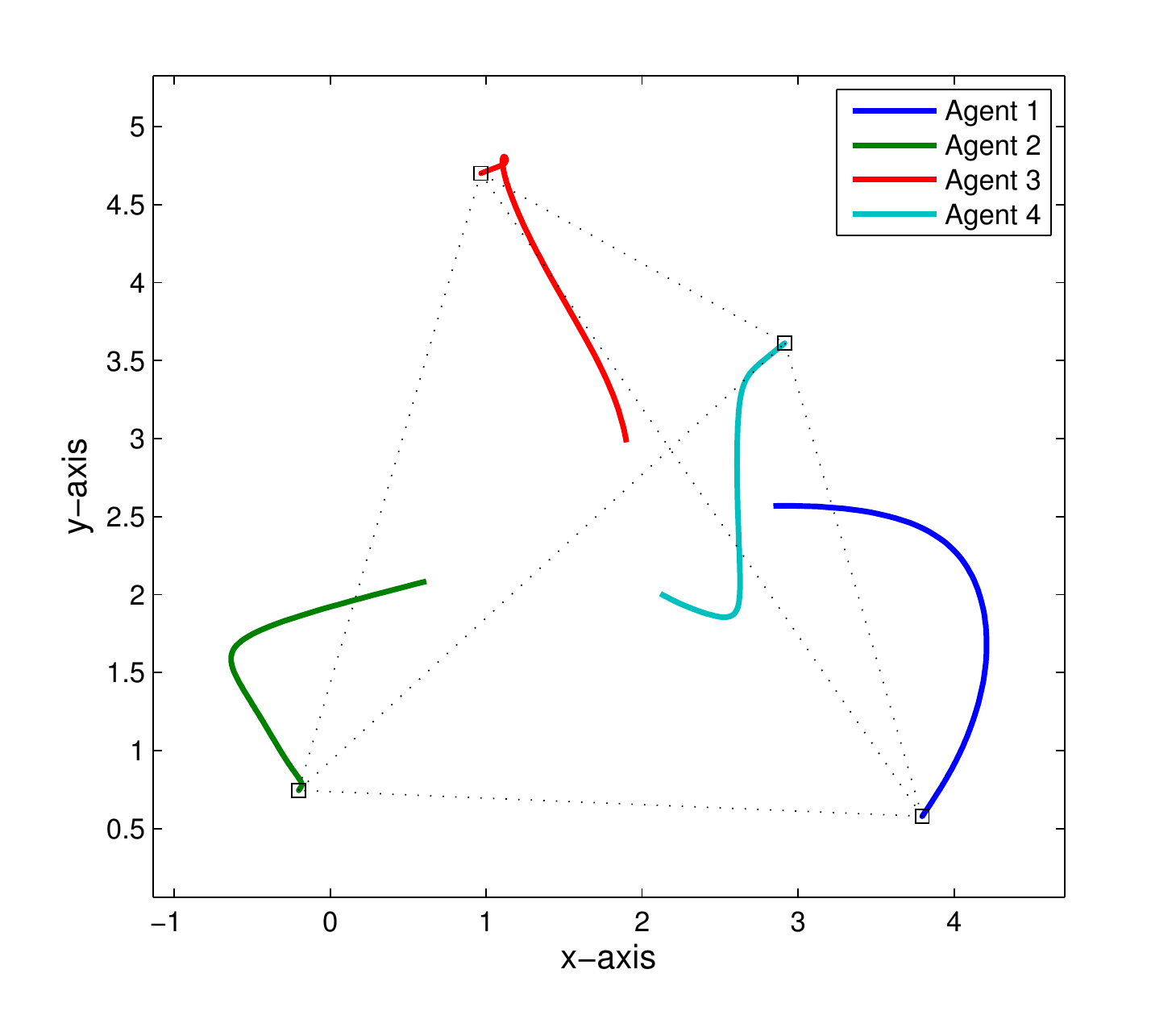}
}
\subfigure[]{
\includegraphics[width=0.48\textwidth]{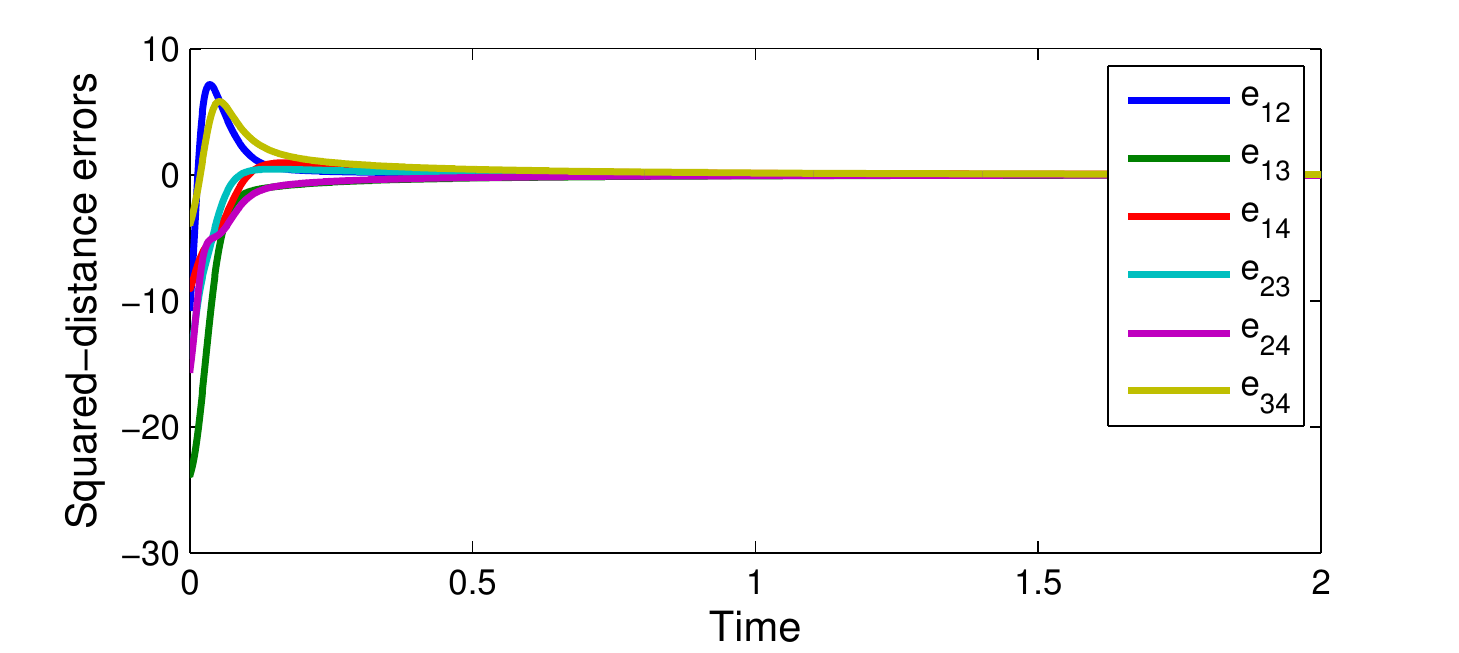}
}
\caption{$\mathcal{K}_{4}$ formation with virtual coordinate.}
\label{Fig:Formation3D}
\end{figure}
Agents 1 through 3 are randomly located on the $x$-$y$ plane at the initial time, and $p_{4z}(0)$ is set to be $\alpha$.
We can observe that the agents finally achieve the desired formation shape, and the distance errors converge to 0 as time goes on.

%%%%%%%%%%%%%%%%%%%%%%%%%%%%%%%%%%%%%%%%%%%%%%%%%%%%%%%%%%%%%%%%%%%%%%%%
\section{Conclusion}
\label{Sec:Conclusion}
In this paper, we proposed a control law to achieve the desired $\mathcal{K}_{4}$ formation in $\mathbb{R}^{2}$.
The proposed control law guarantees almost global convergence in the sense that the formation converges to the desired formation shape for almost all initial conditions.
The proposed control law is one modified from Krick's control law in \cite{C:Krick:IJC2009} with a virtual coordinate.
With the virtual coordinate, the $\mathcal{K}_{4}$ formation in $\mathbb{R}^{2}$ is treated as a virtual tetrahedron formation locked on the plane.
Then, we take advantage of the properties of $\mathcal{K}_{4}$ formations in $\mathbb{R}^{3}$ that can be found in \cite{C:Park:CDC2014, C:Sun:CDC2015} to prove the instability of any incorrect equilibrium formation of the locked tetrahedron formation.

We expect that the proposed variation technique can be extended to general cases of more than four agents.
For example, a formation of five agents might be treated as a virtual 4-D formation.
Depending on the actual realization space, the number of necessary virtual coordinates might be different.
Thus, more rigorous verification on the relationships between actual formations and corresponding virtual formations should be considered.

%%%%%%%%%%%%%%%%%%%%%%%%%%%%%%%%%%%%%%%%%%%%%%%%%%%%%%%%%%%%%%%%%%%%%%%
\section*{Acknowledgment}
This work was supported in part by NICTA Ltd and the Australian Research Council (ARC) under grant DP-160104500.
Z Sun was supported by a Prime Minister's Australia Asia Incoming Endeavour Postgraduate Award.

%%%%%%%%%%%%%%%%%%%%%%%%%%%%%%%%%%%%%%%%%%%%%%%%%%%%%%%%%%%%%%%%%%%%%%%%
\bibliographystyle{IEEEtran}
\bibliography{mybib}

\begin{thebibliography}{10}
\providecommand{\url}[1]{#1}
\csname url@rmstyle\endcsname
\providecommand{\newblock}{\relax}
\providecommand{\bibinfo}[2]{#2}
\providecommand\BIBentrySTDinterwordspacing{\spaceskip=0pt\relax}
\providecommand\BIBentryALTinterwordstretchfactor{4}
\providecommand\BIBentryALTinterwordspacing{\spaceskip=\fontdimen2\font plus
\BIBentryALTinterwordstretchfactor\fontdimen3\font minus
  \fontdimen4\font\relax}
\providecommand\BIBforeignlanguage[2]{{%
\expandafter\ifx\csname l@#1\endcsname\relax
\typeout{** WARNING: IEEEtran.bst: No hyphenation pattern has been}%
\typeout{** loaded for the language `#1'. Using the pattern for}%
\typeout{** the default language instead.}%
\else
\language=\csname l@#1\endcsname
\fi
#2}}

\bibitem{C:Ren:CSM2007}
W.~Ren, R.~W. Beard, and E.~M. Atkins, ``Information consensus in multivehicle
  cooperative control,'' \emph{IEEE Control Systems Magazine}, vol.~27, no.~2,
  pp. 71--82, 2007.

\bibitem{C:OlfatiSaber:IEEE2007}
R.~\mbox{Olfati-Saber}, J.~A. Fax, and R.~M. Murray, ``Consensus and
  cooperation in networked multi-agent systems,'' \emph{Proceedings of the
  IEEE}, vol.~95, no.~1, pp. 215--233, 2007.

\bibitem{C:OlfatiSaber:IFAC2002}
R.~\mbox{Olfati-Saber} and R.~M. Murray, ``Distributed cooperative control of
  multiple vehicle formations using structural potential functions,'' in
  \emph{Proceedings of the 15th IFAC World Congress}, July 2002, pp. 495--500.

\bibitem{C:Krick:IJC2009}
L.~Krick, M.~E. Broucke, and B.~A. Francis, ``Stabilisation of infinitesimally
  rigid formations of multi-robot networks,'' \emph{International Journal of
  Control}, vol.~82, no.~3, pp. 423--439, 2009.

\bibitem{C:Yu:SIAMJCO2009}
C.~Yu, B.~D.~O. Anderson, S.~Dasgupta, and B.~Fidan, ``Control of minimally
  persistent formations in the plane,'' \emph{SIAM Journal on Control and
  Optimization}, vol.~48, no.~1, pp. 206--233, 2009.

\bibitem{C:Oh:Automatica2011}
K.-K. Oh and H.-S. Ahn, ``Formation control of mobile agents based on
  inter-agent distance dynamics,'' \emph{Automatica}, vol.~47, no.~10, pp.
  2306--2312, 2011.

\bibitem{C:Summers:TAC2011}
T.~H. Summers, C.~Yu, S.~Dasgupta, and B.~D.~O. Anderson, ``Control of
  minimally persistent leader-remote-follower and coleader formations in the
  plane,'' \emph{IEEE Transactions on Automatic Control}, vol.~56, no.~12, pp.
  2778--2792, 2011.

\bibitem{C:Anderson:IFAC2010}
B.~D.~O. Anderson, C.~Yu, S.~Dasgupta, and T.~H. Summers, ``Controlling four
  agent formations,'' in \emph{Proceedings of the 2nd IFAC Workshop on
  Distributed Estimation and Control in Networked Systems}, Sept. 2010, pp.
  139--144.

\bibitem{C:Dasgupta:AUCC2011}
S.~Dasgupta, B.~D.~O. Anderson, C.~Yu, and T.~H. Summers, ``Controlling
  rectangular formations,'' in \emph{Proceedings of the 2011 Australian Control
  Conference}, Nov. 2011, pp. 44--49.

\bibitem{C:Park:CDC2014}
M.-C. Park, Z.~Sun, B.~D.~O. Anderson, and H.-S. Ahn, ``Stability analysis on
  four agent tetrahedral formations,'' in \emph{Proceedings of the 53rd IEEE
  Conference on Decision and Control}, Dec. 2014, pp. 631--636.

\bibitem{C:Sun:CDC2015}
Z.~Sun, U.~Helmke, and B.~D.~O. Anderson, ``Rigid formation shape control in
  general dimensions: an invariance principle and open problems,'' in
  \emph{Proceedings of the 54th IEEE Conference on Decision and Control}, Dec.
  2015, pp. 6095--6100.

\bibitem{C:Oh:IJRNC2014}
K.-K. Oh and H.-S. Ahn, ``Distance-based undirected formations of
  single-integrator and double-integrator modeled agents in $n$-dimensional
  space,'' \emph{International Journal of Robust and Nonlinear Control},
  vol.~24, no.~12, pp. 1809--1820, 2014.

\bibitem{C:Summers:MSC2013}
T.~H. Summers, C.~Yu, S.~Dasgupta, and B.~D.~O. Anderson, ``Certifying
  non-existence of undesired locally stable equilibria in formation shape
  control problems,'' in \emph{Proceedings of the 2013 IEEE International
  Symposium on Intelligent Control}, 2013, pp. 200--205.

\bibitem{C:Wirth:EM2009}
K.~Wirth and A.~S. Dreiding, ``Edge lengths determining tetrahedrons,''
  \emph{Elemente der Mathematik}, vol.~64, no.~4, pp. 160--170, 2009.

\bibitem{C:Khalil:Book2002}
H.~K. Khalil, \emph{Nonlinear Systems}, 3rd~ed.\hskip 1em plus 0.5em minus
  0.4em\relax NJ: Prentice Hall, 2002.

\bibitem{C:Anderson:SIAMJCO2014}
B.~D.~O. Anderson and U.~Helmke, ``Counting critical formations on a line,''
  \emph{SIAM Journal on Control and Optimization}, vol.~52, no.~1, pp.
  219--242, 2014.

\end{thebibliography}

\end{document}